\documentclass[11pt]{article}  

\usepackage{amssymb}
\usepackage{amsmath}
\usepackage{amsthm}
\usepackage{color}
\usepackage{colortbl}
\usepackage{graphicx}
\usepackage{hyperref}
\usepackage{fullpage}
\usepackage{mdframed}

\usepackage{lmodern}
\usepackage[T1]{fontenc}
\usepackage[utf8]{inputenc}

\newtheorem{theorem}{Theorem}
\newtheorem{lemma}[theorem]{Lemma}
\newtheorem{definition}[theorem]{Definition}
\newtheorem{fact}[theorem]{Fact}

\newtheorem{remark}[theorem]{Remark}
\newtheorem{corollary}[theorem]{Corollary}

\newcommand{\eps}{\epsilon}
\newcommand{\E}{\mathbb{E}}
\newcommand{\Z}{{\mathbb{Z}}}
\newcommand{\R}{{\mathbb{R}}}

\DeclareMathOperator{\Vol}{Vol}

\newmdenv[linewidth=1]{mystyle}
\begin{document}
\title{Approximate Nearest Neighbor Search in $\ell_p$}
\author{Huy L. Nguy\~{\^{e}}n\\Princeton}
\date{}
\maketitle
\begin{abstract}
We present a new locality sensitive hashing (LSH) algorithm for $c$-approximate nearest neighbor search in $\ell_p$ with $1<p<2$. For a database of $n$ points in $\ell_p$, we achieve $O(dn^{\rho})$ query time and $O(dn+n^{1+\rho})$ space, where $\rho \le O((\ln c)^2/c^p)$. This improves upon the previous best upper bound $\rho\le 1/c$ by Datar et al. (SOCG 2004), and is close to the lower bound $\rho \ge 1/c^p$ by O'Donnell, Wu and Zhou (ITCS 2011). The proof is a simple generalization of the LSH scheme for $\ell_2$ by Andoni and Indyk (FOCS 2006).
\end{abstract}
\section{Introduction}
Approximate nearest neighbor search has been studied extensively in the last few decades. In this problem, a database of $n$ points in $\R^d$ is preprocessed so that given a query point $q$, if the point closest to $q$ in the database is at distance $r$ from $q$, then the algorithm will return a point $p$ in the database within distance $cr$ from $q$. The parameter $c>1$ is the approximation factor of the algorithm. At the moment, the best approach giving good guarantees both in time and space in high dimensions is locality sensitive hashing (LSH)~\cite{HIM12}. The running time and space of LSH-based algorithms depend on a parameter $\rho$, which is determined by the metric space, the approximation factor $c$, and the hash functions: the query time is $dn^{\rho+o(1)}$ and the space is $nd+n^{1+\rho+o(1)}$. For the norm $\ell_p$ with $1\le p\le 2$, it is known that there exists a distribution over hash functions such that $\rho\le 1/c$~\cite{DIIM04}. For only the special case of $\ell_2$, it is known that we can also get $\rho=1/c^2+o(1)$~\cite{AI06}. In~\cite{OWZ11}, it was shown that $\rho\ge c^{-p}$ for all $p\in[1,2]$. In this paper we give a new LSH for $1<p<2$ achieving $\rho = O((\ln c)^2 c^{-p})$. Independently, there is an algorithm~\cite{Indyk13} close to matching the lower bound from~\cite{OWZ11} but we believe it is worthwhile to present the argument here as it is simple and might be applicable elsewhere.

\section{Preliminaries}
Fist we need the formal definition of LSH.
\begin{definition}[\cite{HIM12}]
A family of hash functions $h$ is $(r, cr, p_1, p_2)$-sensitive if
\begin{itemize}
\item If $\|x-y\|_p \le r$ then $\Pr[h(x)=h(y)] \ge p_1$.
\item If $\|x-y\|_p \ge cr$ then $\Pr[h(x)=h(y)] \le p_2$.
\end{itemize}
Define $\rho = \frac{\ln 1/p_1}{\ln 1/p_2}$.
\end{definition}
Given such a hash family for every $r$, one immediately gets an algorithm for approximate nearest neighbor search.
\begin{theorem}[\cite{HIM12}]
If for every $r$, there exists an $(r, cr, p_1, p_2)$-sensitive hash family with the parameter $\rho$ uniformly bounded from above by $\rho_0$, evaluation time $dn^{o(1)}$, and $1/p_1 = n^{o(1)}$, then there is an algorithm for finding $c$-approximate nearest neighbor with query time $dn^{\rho_0+o(1)}$ and space $O(dn)+n^{1+\rho_0+o(1)}$.
\end{theorem}

The rest of the paper focuses on analyzing $\rho$ for a fixed $r$. Because we can always scale all distances, assume wlog that $r=1$.

Let $B_p(x,r)$ denote the $\ell_p$ ball of radius $r$ centered at $x$. Let $V_t$ be the volume of $B_p(\vec{0}, w)$ in $\R^t$. Let $L_t$ denote the lattice $\{\sum_{i=1}^t \Delta a_i w e_i\ |\ a_i\in \Z\}$, where $e_i$ is the $i$th standard basis vector, $\Delta = 4$, and $w=O(c\ln c)$.

\section{The Hash Function}
The hash function works in a similar way to~\cite{AI06}, with some modifications to the parameters. First it uses the $p$-stable distribution to reduce the dimension to $t=\Theta((c\ln c)^{p})$. Then, it partitions the $t$-dimensional space using lattices of balls of radius $w=O(c\ln c)$. See Figure~\ref{fig:alg} for details.
\begin{figure}[ht]
\begin{mystyle}
\paragraph{Choosing a hash function $h\in \mathcal{H}$}
\begin{enumerate}
\item For each $u\in \{1,2,\ldots, U_t\}$, pick a random shift $s_u\in [0, \Delta w]^t$ to specify the shifted lattice $s_u+ L_t$.
\item Pick a random matrix $A\in \R^{t\times d}$ whose entries are i.i.d. $p$-stable random variables with the scale parameter 1. Let $A'=T^{-1/p}A$ ($T$ is the threshold defined in Lemma~\ref{lem:concentration}).
\end{enumerate}
\paragraph{Applying the hash function $h$ to a point $x\in \R^d$}
\begin{enumerate}
\item Let $x' = A'x$.
\item Find the smallest $u\in \{1,2, \ldots, U_t\}$ such that there exists a point $y\in L_t$ satisfying $x\in B_p(y+s_u, w)$. If $u$ exists then the hash value of $x$ is the pair $(u, y)$. Otherwise, the hash value of $x$ is the pair $(0,\vec{0})$.
\end{enumerate}
\end{mystyle}
\caption{\label{fig:alg}The algorithm for computing the hash value of a given point $x\in \R^d$.}
\end{figure}
\section{Analysis}
First, we bound the number of lattices of balls needed to cover the entire space $\R^t$. This number determines the running time of the hash function as finding the closest ball in a lattice to a given point is simple: one just needs to find the closest lattice point in each coordinate separately. While this operation uses the floor function, we believe the usage is justified as the coordinates do not encode special information and it is also widely used in the LSH literature. The following lemma is a generalization of~\cite[Lemma 3.2.2]{Andoni09} with an analogous proof.
\begin{lemma}[\cite{Andoni09}]\label{lem:cover}
Consider the $t$-dimensional space $\R^t$ and let $\delta$ be a positive constant. Let $B_{u}$ be the collection of $\ell_p$ balls centered at lattice points of $s_u+L_t$, where $s_u$ is a uniformly random vector in $[0,\Delta w]^t$. If $U_t = \Delta^t t^{t/p+1}\ln (\Delta t/\delta)$ then the collections $B_1, \ldots, B_{U_t}$ cover all of $\R^t$ with probability at least $1-\delta$.
\end{lemma}
\begin{proof}
The proof is a standard covering argument. We present it here for completeness. By the regularity of the lattices, the whole space is covered iff the cube $[0,\Delta w]^t$ is covered. Divide the cube into subcubes of side length $w/t^{1/p}$. If some $s_u$ lies in a subcube then the whole subcube is covered. The probability that some $s_u$ lies in a particular cube is $1/N$, where $N$ is the number of subcubes. We have $N=(\Delta t^{1/p})^t$. If $U_t \ge N\ln (N/\delta)$ then by the union bound, the probability that some subcube is not covered is bounded by
$$N(1-1/N)^{U_t} \le \exp(\ln N- U_t/N) = \exp(-\ln 1/\delta) = \delta$$
\end{proof}

\begin{corollary}
For $\delta=\exp(-\Theta(t))$, the time to evaluate the hash function is $dc^{O((c\ln c)^p)}$.
\end{corollary}

To analyze the first step of the hash function, we need a concentration bound for $p$-stable distribution. The proof is similar to that of a similar bound for $p=1$ by~\cite{Indyk06}.

\begin{lemma}\label{lem:concentration}
Let $p$ be a constant in $(1,2)$. Let $x\in \R^ d$ and a random matrix $A\in \R^{t\times d}$ whose entries are i.i.d. $p$-stable random variables with the scale parameter 1. For $t\rightarrow \infty$, there exists a threshold $T= T(t, \eps)$ such that
\begin{itemize}
\item $\Pr[\|Ax\|_p^p < T\|x\|_p^p] \le \exp(-\Theta(t^{1-\eps p}(\eps \ln t)^2))$
\item $\Pr[\|Ax\|_p^p > 2^{(4+p)/2}\eps^{-1}T\|x\|_p^p] \le 1/2$
\end{itemize}
\end{lemma}

\begin{proof}
First, we need an approximation of the probability density function of the $p$-stable distribution. We use the following theorem from~\cite[Theorem 42]{Nelson11}.
\begin{theorem}[\cite{Nelson11}]\label{thm:pdf}
Define $\phi_p^+$ and $\phi_p^-$ as follows:
$$\phi_p^-(x) = \frac{a}{x^{p+1}} - \frac{b}{x^3}$$
and
$$\phi_p^+(x) = \frac{2^{(p+1)/2}a}{x^{p+1}} + \frac{b}{x^3}$$
for certain constants $a,b$. Then for any $x\ge 1$, $\phi_p^-(x) \le \phi_p(x) \le \phi_p^+(x)$.
\end{theorem}

Since $x\rightarrow Ax$ is a linear map, we can assume wlog that $\|x\|_p=1$. Each coordinate of $Ax$ is an i.i.d. $p$-stable random variable with the scale parameter $1$. Let $Y_i$ denote the absolute value of the $i$th coordinate of $Ax$. Define $Z_{i, M}:=\min(Y_i, M)$. We first prove some properties of $Z_{i,M}$.

\begin{lemma} For $M\rightarrow \infty$, we have
\begin{align*}
\E[Z_{i, M}^p] &= \Theta(\ln M)\\
\E[Z_{i, M}^{2p}] &= \Theta(M^p)
\end{align*}
and in particular, $\E[Z_{i, M^{1/\eps}}^p] \le (2^{(1+p)/2}+o(1)) \E[Z_{i,M}^p]/\eps$.
\end{lemma}
\begin{proof}
Let $p_M$ be the probability that a standard $p$-stable random variable exceeds $M$. We can bound $p_M$ by
$$p_M \le \int_{M}^{\infty} \phi_p^+(x) dx = \frac{2^{(p+1)/2}a}{pM^p}+ \frac{b}{2M^2}$$
and
$$p_M \ge \int_{M}^{\infty} \phi_p^-(x) dx = \frac{a}{pM^p}- \frac{b}{2M^2}$$

First we bound $\E[Z_{i, M}^p]$.
\begin{align*}
\E[Z_{i,M}^p] &\le \int_{0}^{1} dx + \int_{1}^{M} x^p \phi_p^{+}(x) dx + p_M M^p\\
&\le O(1) + 2^{(p+1)/2}a\ln M
\end{align*}
Similarly
\begin{align*}
\E[Z_{i,M}^p] &\ge \int_{1}^{M} x^p \phi_p^{-}(x) dx\\
&\ge a\ln M - O(1)
\end{align*}

Next we bound $\E[Z_{i,M}^{2p}]$.
\begin{align*}
\E[Z_{i,M}^{2p}] &\le \int_{0}^{1} dx + \int_{1}^{M} x^{2p} \phi_p^{+}(x) dx + p_M M^{2p}\\
&\le (2^{(p+3)/2}a/p+o(1))M^p
\end{align*}
Similarly
\begin{align*}
\E[Z_{i,M}^{2p}] &\ge \int_{1}^{M} x^{2p} \phi_p^{-}(x) dx\\
&\ge (a/p-o(1))M^p
\end{align*}
\end{proof}

Set $M = t^{\eps}$ and $T=t\E[Z_{i,M}^p]/2 = \Theta(\eps t\ln t)$. We have $\Pr[\|Ax\|_p^p < T] \le \Pr[\sum_i Z_{i,M}^p < T]$. By an inequality by Maurer~\cite{Maurer03}.
$$\Pr[\sum_i Z_{i,M}^p < T] \le \exp\left(-\frac{T^2}{2\sum_i \E[Z_{i,M}^{2p}]}\right)=\exp(-\Theta(t^{1-\eps p}(\eps\ln t)^2))$$

On the other hand,
\begin{align*}
\Pr[\sum_i Y_i^p > 2^{(4+p)/2}\eps^{-1}T] &\le \Pr[\exists i:Y_i\ge M^{1/\eps}] + \Pr[\sum_i Y_i^p > 2^{(4+p)/2}\eps^{-1}T | \forall i: Y_i < M^{1/\eps}]\\
&\le t(2^{(p+1)/2}a/(pM^{p/\eps}) + b/(2M^{2/\eps})) + \frac{\E[\sum_i Y_i^p | \forall i: Y_i < M^{1/\eps}]}{2^{(4+p)/2}\eps^{-1}T}\\
&\le O(t^{1-p}) + \frac{\E[\sum_i Z_{i,M^{1/\eps}}^p]}{2^{(4+p)/2}\eps^{-1}T}\\
&\le O(t^{1-p}) + \frac{(2^{(1+p)/2}+o(1))\eps^{-1}T}{2^{(4+p)/2}\eps^{-1}T} < 1/2
\end{align*}
\end{proof}

To analyze the second step of the hash function, we use the uniform convexity and smoothness properties of $\ell_p$, see e.g.~\cite{BCL94}.
\begin{fact}
For any $1< p\le 2$,
\begin{itemize}
\item $\ell_p$ is $p$-uniformly smooth:
\begin{equation}\label{eq:smooth}
\forall x,y\in\ell_p,~\frac{\|x\|_p^p+\|y\|_p^p}{2} \le \left\|\frac{x+y}{2}\right\|_p^p + \left\|\frac{x-y}{2}\right\|_p^p
\end{equation}
\item $\ell_p$ is $2$-uniformly convex:
\begin{equation}\label{eq:convex}
\forall x,y\in\ell_p,~\frac{\|x\|_p^2+\|y\|_p^2}{2} \ge \left\|\frac{x+y}{2}\right\|_p^2 + (p-1)\left\|\frac{x-y}{2}\right\|_p^2
\end{equation}
\end{itemize}
\end{fact}

Finally we are ready to prove the main technical lemma determining the parameter $\rho$. It can be viewed as a generalization of~\cite[Lemma 3.2.3]{Andoni09}. Before proceeding to the lemma, we want to note that conditioned on the fact that the whole space $\R^t$ is covered by the lattices (which happens with high probability by Lemma~\ref{lem:cover}), for any two point $x, y\in\R^t$, the probability that they are contained in the same ball in the partition of $\R^t$ defined by the shifted lattices of $h$ is exactly $\frac{\Vol(B_p(x, w) \cap B_p(y,w))}{\Vol(B_p(x,w) \cup B_p(y,w))}$. Removing the conditioning only changes the collision probabilities by at most $\delta=\exp(-\Theta(t))$, which is negligible. In a nutshell, the proof combines two observations. First, by Lemma~\ref{lem:concentration}, the mapping $x\rightarrow Ax$ does not distort distances by a large amount. Second, for points in $\R^t$, the volumes involved in collision probabilities can be approximated by volumes of balls of different radii. Because the ratio of volumes of balls of different radii can easily be computed from the ratio of the radii, we can approximate the collision probabilities.

\begin{lemma}\label{lem:rho}
Let $p$ be a constant in $(1,2)$. Let $x,y$ be two points in $\R^d$. Let $p_1$ be the collision probability when $\|x-y\|_p\le 1$ and $p_2$ be the collision probability when $\|x-y\|_p\ge c$. Then, for $w=\Theta(c\ln c), t=\Theta(w^p)$, we have $\rho = \frac{\ln p_1}{\ln p_2} = O((\ln c)^2 c^{-p})$ as $c\rightarrow \infty$.
\end{lemma}

\begin{proof}
Let $x' = A'x, y' = A'y$. We first analyze the volume of $A_1 = B_p(x',w) \cap B_p(y',w)$ when $\|x'-y'\|_p \le r$. We will show 
\begin{align}
\nonumber B_p((x'+y')/2, w(1-(2+\gamma)r^p/(2w)^p)^{1/p})\subset A_3 = A_1&\cup B_p(x', w(1-(2+2\gamma)r^p/(2w)^p)^{1/p})\\
&\cup B_p(y', w(1-(2+2\gamma)r^p/(2w)^p)^{1/p})\label{eq:vol-lb}
\end{align}
for arbitrary $\gamma>0$. Setting $\gamma$ close to $0$ results in a better constant in the final bound of $\rho$ but for ease of understanding, we can set $\gamma=1$. Consider a point $z\in B_p((x'+y')/2, w(1-(2+\gamma)r^p/(2w)^p)^{1/p}) \setminus A_3$. Wlog, we assume $\|x'-z\|_p \ge \|y'-z\|_p$. By the assumptions, we have 
\begin{align*}
\|(x'+y')/2-z\|_p &\le w\left(1-\frac{(2+\gamma)r^p}{(2w)^p}\right)^{1/p}\\
\|x'-z\|_p &> w\\
\|y'-z\|_p &> w\left(1-\frac{(2+2\gamma)r^p}{(2w)^p}\right)^{1/p}
\end{align*}
Applying~\ref{eq:smooth} to $x'-z$ and $y'-z$, we have:
\begin{align*}
\left\|\frac{x'-y'}{2}\right\|_p^p &\ge \frac{\|x'-z\|_p^p+\|y'-z\|_p^p}{2} - \left\|\frac{x'+y'}{2}-z\right\|_p^p\\
&> \frac{w^p+w^p(1-(2+2\gamma)r^p/(2w)^p)}{2} - w^p(1-(2+\gamma)r^p/(2w)^p)\\
&\ge (r/2)^p
\end{align*}
This contradicts the assumption that $\|x'-y'\|_p \le r$. In other words, there is no such point $z$ and $B_p((x+y)/2, w(1-(2+\gamma)r^p/(2w)^p)^{1/p})\subset A_3$. Note that for any $\alpha\in \R$ and $u\in\R^t$, we have $\Vol(B_p(u, \alpha w)) = \alpha^t V_t$. Applying this fact to~(\ref{eq:vol-lb}), we have
\begin{align*}
\Vol(A_1) &\ge V_t \left(\left(1-\frac{(2+\gamma)r^p}{(2w)^p}\right)^{t/p} - 2 \left(1-\frac{(2+2\gamma)r^p}{(2w)^p}\right)^{t/p}\right)\\
&\ge V_t \left(1-\frac{(2+\gamma)r^p}{(2w)^p}\right)^{t/p}\left(1 - 2 \left(1-\frac{\gamma r^p/2}{(2w)^p}\right)^{t/p}\right)\\
&\ge V_t\left(1-\frac{(2+\gamma)r^p}{(2w)^p}\right)^{t/p} (1-\exp(-\Omega(\gamma t r^p w^{-p})))
\end{align*}
By Lemma~\ref{lem:concentration},  when $\|x-y\|_p \le 1$, with probability at least $1/2$, we have $\|x'-y'\|_p=O(\eps^{-1})$. Therefore we get an upper bound for $\ln(1/p_1)$,
\begin{align*}
\ln(1/p_1) &\le \ln \left(2 \cdot \frac{2V_t-\Vol(A_1)}{\Vol(A_1)}\right)\\
&\le \ln 4 - \ln(\Vol(A_1))\\
&\le \ln 4 + \frac{O(2+\gamma)\cdot t/(p (2\eps w)^p)}{1-O(2+\gamma)\cdot 1/(2\eps w)^p} + \exp(-\Omega(\gamma t\eps^{-p}w^{-p}))\\
&\le O(2+\gamma)\cdot t/(p(2\eps w)^p)
\end{align*}
The second to last inequality follows from the inequality $\ln(1-x) \ge -x/(1-x)~\forall x\in [0,1)$.

Next, we analyze the volume of $A_2 = B_p(x',w) \cap B_p(y',w)$ when $\|x'-y'\|_p \ge c$. We will show $A_2\subset B_p((x'+y')/2, w\sqrt{1-(p-1)(c/w)^2/4})$. Let $z$ be an arbitrary point in $A_2$. Applying~\ref{eq:convex} to $x'-z$ and $y'-z$, we have:
$$\left\|\frac{x'+y'}{2}-z\right\|_p^2 \le \frac{\|x'-z\|_p^2+\|y'-z\|_p^2}{2} - (p-1)\left\|\frac{x'-y'}{2}\right\|_p^2\le w^2-(p-1)(c/2)^2$$
Thus,
$$\Vol(A_2) \le V_t (1-(p-1)(c/w)^2/4)^{t/2}$$
By Lemma~\ref{lem:concentration},  when $\|x-y\|_p \ge c$, with probability at least $1-P = 1-\exp(-\Theta(t^{1-\eps p}(\eps \ln t)^2))$, we have $\|x'-y'\|_p\ge c$. Therefore, we get a lower bound for $\ln (1/p_2)$,
\begin{align*}
\ln(1/p_2) \ge& \ln\frac{2V_t-\Vol(A_2)}{\Vol(A_2) + (2V_t-\Vol(A_2))P}\\
\ge&\ln\frac{1-\Vol(A_2)/(2V_t)}{\max(\Vol(A_2)/V_t, 2P)}\\
\ge& \ln(1/2) +\min((p-1)tc^2/(8w^2), -\ln 2 P)\\
\ge& (1-o(1))(p-1)tc^2/(8w^2)
\end{align*}
when $(p-1)tc^2/(8w^2) < -\ln 2 P = \Theta(t^{1-\eps p}(\eps \ln t)^2)-\ln 2$. In other words, we need $w/c = \Omega(t^{\eps p/2}/(\eps \ln t))$. Combining the bounds for $p_1$ and $p_2$, we have

$$\rho = \frac{\ln (1/p_1)}{\ln (1/p_2)} \le \frac{O(2+\gamma)\cdot (w/c)^{2-p}}{\eps^p 2^p p(p-1)c^p}$$

We get the stated bound by choosing $w=\Theta(c\ln c), t=\Theta(w^p), \eps=\Theta(\ln \ln t/\ln t)$.
\end{proof}
\begin{remark}
It is possible to slightly tighten the bound by setting $w=\Theta(c), t=\Theta(w^p), \eps=\Theta(1/\ln t)$. The constant $2^{(4+p)/2}$ in Lemma~\ref{lem:concentration} becomes a larger constant depending on the constants in Theorem~\ref{thm:pdf}, but the rest of the proof remains the same. This setting gives $\rho=O((\ln c/c)^p)$, where the O hides a constant depending on the constants in Theorem~\ref{thm:pdf}.
\end{remark}
\section{Discussion}
The second half of the argument uses only the uniform smoothness and convexity properties of the norm while the first half is tailored to $\ell_p$. This leads to the question of whether one can generalize the argument here to get an algorithm for approximate nearest neighbor search for a more general class of norms.
\section{Acknowledgments}
We thank Jelani Nelson and Ilya Razenshteyn for helpful comments.
\bibliographystyle{alpha}
\bibliography{lp-nns}
\end{document}